\documentclass[11pt]{llncs}
\usepackage[margin=1.3in]{geometry} 
\usepackage{comment}
\usepackage{amssymb}
\usepackage{amsmath}
\usepackage{enumerate}
\usepackage{mathtools}
\usepackage{tkz-berge}
\usepackage{tikz}
\usetikzlibrary{positioning,arrows,decorations.pathreplacing,shapes}
\tikzset{>=latex}

\newcommand\tw{\operatorname{tw}}

\newcommand{\poly}{\operatorname{poly}}
\newcommand{\Gpoly}{\mathcal{G}_{\poly}}
\newcommand{\cO}{\mathcal{O}}
\newcommand{\pmc}{potential maximal clique}

\title{On \textsc{Distance-$d$ Independent Set} and other problems in graphs with ``few'' minimal separators}
\author{Pedro Montealegre  \and Ioan Todinca}
\institute{
Univ. Orl\'{e}ans, INSA Centre Val de Loire, LIFO EA 4022,  Orl{\'e}ans , France
\\ \texttt{(pedro.montealegre $\mid$ ioan.todinca)@univ-orleans.fr}
}

\begin{document}
\maketitle
\begin{abstract}
Fomin and Villanger (\cite{FoVi10}, STACS 2010) proved that \textsc{Maximum Independent Set}, \textsc{Feedback Vertex Set}, and more generally the problem of finding a maximum induced subgraph of treewith at most a constant $t$, can be solved in polynomial time on graph classes with polynomially many minimal separators. We extend these results in two directions. Let $\Gpoly$ be the class of graphs with at most $\poly(n)$ minimal separators, for some polynomial $\poly$.

We show that the odd powers of a graph $G$ have at most as many minimal separators as $G$. Consequently, \textsc{Distance-$d$ Independent Set}, which consists in finding maximum set of vertices at pairwise distance at least $d$, is polynomial on $\Gpoly$, for any even $d$. The problem is NP-hard on chordal graphs for any odd $d \geq 3$~\cite{EGM14}.

We also provide polynomial algorithms for \textsc{Connected Vertex Cover} and \textsc{Connected Feedback Vertex Set} on subclasses of $\Gpoly$ including chordal and circular-arc graphs, and we discuss variants of independent domination problems.
\end{abstract}


\section{Introduction}

Several natural graph classes are known to have polynomially many minimal separators, w.r.t. the number $n$ of vertices of the graph. It is the case for \emph{chordal} graphs, which have at most $n$ minimal separators~\cite{RTL76}, \emph{weakly chordal}, \emph{circular-arc} and \emph{circle} graphs, which have $\cO(n^2)$ minimal separators~\cite{BoTo01,KKW98}.

The property of having polynomially many minimal separators has been used in algorithms for decades, initially in an ad-hoc manner, i.e., algorithms were based on minimal separators but also other specific features of particular graph classes (see, e.g., ~\cite{BKK93,KKW98}). Later, it was observed that minimal separators are sufficient for solving problems like \textsc{Treewidth} or  \textsc{Minimum fill-in}~\cite{BoTo01,BoTo02}. Both problems are related to \emph{minimal triangulations}. Given an arbitrary graph $G$, a minimal triangulation is a minimal chordal supergraph $H$ of $G$, on the same vertex set. Bouchitt\'e and Todinca~\cite{BoTo01} introduced the notion of \emph{\pmc}, that is, a vertex set of $G$ inducing a maximal clique in some minimal triangulation $H$ of $G$. Their algorithm for treewidth is based on dynamic programming over minimal separators and \pmc s. The same authors proved that the number of \pmc s is polynomially bounded in the number of minimal separators~\cite{BoTo02}. 

Fomin and Villanger~\cite{FoVi10} found a more surprising application of minimal separators and \pmc s, proving that they were sufficient for solving problems like \textsc{Maximum Independent Set},  \textsc{Maximum Induced Forest}, and more generally for finding a maximum induced subgraph $G[F]$ of treewidth at most $t$, where $t$ is a constant.

More formally, let $\poly$ be some polynomial. We call $\Gpoly$ the family of graphs  such that $G \in \mathcal{G}_{\poly}$ if and only if $G$ has at most $\poly(n)$ minimal separators. By~\cite{FoVi10}, the problem of finding a maximum induced subgraph of treewidth at most $t$ can be solved in polynomial time on $\Gpoly$. The exponent of the polynomial depends on $\poly$ and on $t$. In~\cite{FTV15}, Fomin \textit{et al.} further extend the technique to compute large induced subgraphs of bounded treewidth, and satisfying some CMSO property (expressible in counting monadic second-order logic). That allows to capture problems like \textsc{Longest induced path}. They also point out some limits of the approach. It is asked in~\cite{FTV15} whether the techniques can be extended for solving the \textsc{Connected Vertex Cover} problem, which is equivalent to finding a maximum independent set $F$ such that $G - F$ is connected. More generally, their algorithm computes an induced subgraph $G[F]$ of treewidth at most $t$ satisfying some CMSO property, but is not able to ensure any property relating the induced subgraph to the initial graph. 


Here we make some progress in this direction.
First, we consider the problem \textsc{Distance-$d$ Independent Set} on $\Gpoly$, where the goal is to find a maximum independent set $F$ of the input graph $G$, such that the vertices of $F$ are at pairwise distance at least $d$ in $G$ (in the literature this problem is also known as {\sc $d$-Scattered-Set}). This is equivalent to finding a maximum independent set in graph $G^{d-1}$, the $(d-1)$-th power of $G$. Eto \textit{et al.}~\cite{EGM14} already studied the problem on chordal graphs, and proved that it is polynomial for every even $d$, and NP-hard for any odd $d\geq 3$ (it is even $W[1]$-hard when parameterized by the solution size). Their positive result is based on the observation that for any even $d$, if $G$ is chordal then so is $G^{d-1}$. Eto \textit{et al.}~\cite{EGM14} ask if \textsc{Distance-$d$ Independent Set} is polynomial on chordal bipartite graphs (which are \emph{not} chordal but weakly chordal, see Section~\ref{se:prelim}), a subclass of $\Gpoly$. We bring a positive answer to their question for even values $d$, by a result of combinatorial nature: for any graph $G$ and any odd $k$, the graph $G^k$ has no more minimal separators than $G$ (see Section~\ref{se:pow}). Consequently,  \textsc{Distance-$d$ Independent Set} is polynomial on $\Gpoly$, for any even value $d$ and any polynomial $\poly$, and NP-hard for any odd $d \geq 3$ and any $\poly(n)$ asymptotically larger than $n$. Such a dichotomy between odd and even values also appears when computing large $d$-clubs, that are induced subgraphs of diameter at most $d$~\cite{GHKR14}, and for quite similar reasons.

Second, we consider \textsc{Connected Vertex Cover}, \textsc{Connected Feedback Vertex Set} and more generally the problem of finding a maximum induced subgraph $G[F]$ of treewidth at most $t$, such that $G - F$ is connected. We show (Section~\ref{se:cvc}) that the problems are polynomially solvable for subclasses of $\Gpoly$, like chordal and circular-arc graphs. 
This does not settle the complexity of these problems on $\Gpoly$. As we shall discuss in Section~\ref{se:ind}, when restricted to bipartite graphs in $\Gpoly$, \textsc{Connected Vertex Cover} can be reduced from \textsc{Red-Blue Dominating Set} (see~\cite{DLS14}). It might be that this latter problem is NP-hard on bipartite graphs of $\Gpoly$; that was our hope, since the very related problem \textsc{Independent Dominating Set} is NP-hard on chordal bipartite graphs~\cite{DMK90}, and on circle graphs~\cite{BGMPST14}. This question is still open, however we will observe that the \textsc{Red-Blue Dominating Set} is polynomial on the two natural classes of bipartite graphs with polynomially many minimal separators: chordal bipartite and  circle bipartite graphs.

\section{Preliminaries}\label{se:prelim}

Let $G = (V,E)$ be a graph.  Let $dist_G(u,v)$ denote the distance between vertices $u$ and $v$ (the minimum number of edges of a $uv$-path). We denote by $N_G^k[v]$ the set of vertices at distance at most $k$ from $v$. Let also $N_G^k(v) = N_G^k[v] \setminus \{v\}$, and we call these sets the \emph{closed} and \emph{open neighborhoods at distance} $k$ of $v$, respectively.  Similarly, for a set of vertices $U \subseteq V$, we call the sets $N_G^k(U) = \cup_{u\in U} N_G^k(u) \backslash~ U$ and $N_G^k[U] = \cup_{u\in U} N_G^k[u]$ the open and closed neighborhoods at distance $k$ of $U$, respectively. For $k=1$, we simply denote by $N_G(U)$, respectively $N_G[U]$, the open and closed neighborhoods of $U$; the subscript is omitted if clear from the context.

A \emph{clique} (resp. \emph{independent set}) of $G$ is a set of pairwise adjacent (resp. non-adjacent) vertices. A distance-$d$ independent set is a set of vertices at pairwise distance at least $d$. Equivalently, it is an independent set of the $(d-1)$-th power $G^{d-1}$ of $G$. Graph $G^k = (V,E^k)$ is obtained from $G$ by adding an edge between every pair of vertices at distance at most~$k$.

Given a vertex subset $C$ of $G$, we denote by $G[C]$ the subgraph induced by $C$. We say that $C$ is a connected component of $G$ if $G[C]$ is connected and $C$ is inclusion-maximal for this property. For $S \subseteq V$, we simply denote $G - S$ the graph $G[V \setminus S]$. We say that $S$ is a \emph{$a,b$-minimal separator} of $G$ if $a$ and $b$ are in  distinct components $C$ and $D$ of $G-S$, and $N(C) = N(D) = S$. We also say that $S$ is a minimal separator if it is an $a,b$-minimal separator for some pair of vertices $a$ and $b$.

\begin{proposition}[\cite{BBH00}]\label{prop:min_sep}
Let $G = (V,E)$ be a graph,  $C$ be a connected set of vertices, and let $D$ be a component of $G - N[C]$. Then $N(D)$ is an $a,b$-minimal separator of $G$, for any $a \in C$ and $b \in D$. 
\end{proposition}

\subsection{Graph classes}

A graph is \emph{chordal} if it has no induced cycle with more than three vertices. A graph $G$ is \emph{weakly chordal} if $G$ and its complement $\overline{G}$ have no induced cycle with more than four vertices. 

The classes of \emph{circle} and \emph{circular-arc graphs} are defined by their intersection model. A graph $G$ is a \emph{circle graph} (resp. a \emph{circular-arc graph}) if every vertex of the graph can be associated to a chord (resp. to an arc) of a circle such that two vertices are adjacent in $G$ if and only if the corresponding chords (resp. arcs) intersect. We may assume w.l.o.g. that, in the intersection model, no two chords (resp.  no two arcs) share an endpoint. On the circle, we add a \emph{scanpoint} between each two consecutive endpoints of the set of chords (resp. arcs). A \emph{scanline} is a line segment between two scanpoints. Given an intersection model of a circle (resp. circular-arc) graph $G$, for any minimal separator $S$ of $G$ there is a scanline such that the vertices of $S$ correspond exactly to the chords (resp. arcs) intersecting the scanline, see, e.g.,~\cite{KKW98}. 

Chordal graphs have at most $n$ minimal separators~\cite{RTL76}; weakly chordal, circle and circular-arc graphs all have $\cO(n^2)$ minimal separators~\cite{BoTo01,KKW98}.

\begin{definition}\label{de:gpoly}
Let $\poly$ be some polynomial. We call $\mathcal{G}_{\poly}$ the family of graphs  such that $G \in \mathcal{G}_{\poly}$ if and only if $G$ has at most $\poly(n)$ minimal separators, where $n = |V(G)|$.
\end{definition}

\subsection{Dynamic programming over minimal triangulations}\label{ss:FoVi10}

Let $G=(V,E)$ be an arbitrary graph. A chordal supergraph $H = (V, E')$ (i.e., with $E \subseteq E'$), is called a \emph{triangulation} of $G$. If, moreover, $E'$ is inclusion-minimal among all possible triangulations, we say that $H$ is a \emph{minimal triangulation} of $G$.

The \emph{treewidth} of a chordal graph is its maximum clique size, minus one. Forests have treewidth $1$, and graphs with no edges have treewidth $0$. The treewidth $\tw(G)$ of an arbitrary graph $G$ is the minimum treewidth over all (minimal) triangulations $H$ of $G$.

Cliques of minimal triangulations play a central role in treewidth. A \emph{potential maximal clique} of $G$ is a set of vertices that induces a maximal clique in some minimal triangulation $H$ of $G$. By~\cite{BoTo01}, if $\Omega$ is a \pmc, then for every component $C_i$ of $G - \Omega$, its neighborhood $S_i$ is a minimal separator. Moreover, the sets $S_i$ are exactly the minimal separators of $G$ contained in $\Omega$.

\begin{proposition}[\cite{BBC00,BoTo02}]
For any polynomial $\poly$, there is a polynomial-time algorithm enumerating the minimal separators and the \pmc s of graphs on $\Gpoly$.
\end{proposition}

Minimal separators and \pmc s have been used for computing treewidth and other parameters related to minimal triangulations, on $\Gpoly$. Fomin and Villanger~\cite{FoVi10} extend the techniques to a family of problems:

\begin{proposition}[\cite{FoVi10}]\label{pr:FoVi10}
For any polynomial $\poly$ and any constant $t$, there is a polynomial algorithm computing a \textsc{Maximum Induced Subgraph of Treewidth at most $t$} on $\Gpoly$.
\end{proposition}

Clearly, \textsc{Maximum Independent Set} (which is equivalent to \textsc{Minimum Vertex Cover}) and \textsc{Maximum Induced Forest} (which is equivalent to \textsc{Minimum Feedback Vertex Set}) fit into this framework: they consist in finding maximum induced subgraphs $G[F]$ of treewidth at most 0, respectively at most 1. 
The first ingredient of~\cite{FoVi10} is the following observation.

\begin{proposition}[\cite{FoVi10}]\label{pr:compat}
Let $G = (V,E)$ be a graph, $F \subseteq V$, and let $H_F$ be a minimal triangulation of $G[F]$. There exists a minimal triangulation $H_G$ of $G$ such that $H_G[F] = H_F$. 
We say that $H_G$ \emph{respects} the minimal triangulation $H_F$ of $G[F]$.
\end{proposition}
Note that, for any clique $\Omega$ of $H_G$, we have that $F \cap \Omega$ induces a clique in $H_F$. In particular, if $\tw(G[F]) \leq t$ and the clique size of $H_F$ is at most $t+1$, then every maximal clique of $H_G$ intersects $F$ in at most $t+1$ vertices.

The second ingredient is a dynamic programming scheme that we describe below. Let $S$ be a minimal separator of $G$, and $C$ be a component of $G - S$ such that $N(C) = S$. The pair $(S,C)$ is called a \emph{block}. Let $\Omega$ be a \pmc\ such that $S \subset \Omega \subseteq S \cup C$. Then $(S,C,\Omega)$ is called a \emph{good triple}.
In the sequel, $W$ denotes a set of at most $t+1$ vertices. 

\begin{definition}\label{de:partcomp}
Let $(S,C)$ (resp. $(S,\Omega,C)$) be a block (resp. a good triple) and let $W \subseteq S$ (resp. $W \subseteq \Omega$) be a set of vertices of size at most $t+1$ . We say that a vertex set $F$ is a \emph{partial solution compatible with $(S,C,W)$} (resp. with $(S,C,\Omega,W)$) if:
\begin{enumerate}
\item $G[F]$ is of treewidth at most $t$,
\item $F \subseteq S \cup C$,
\item $W = F \cap S$ (resp. $W = F \cap \Omega$),
\item there is a minimal triangulation $H$ of $G$ respecting some minimal triangulation of $G[F]$ of treewidth at most $t$, such that $S$ is a minimal separator (resp. $S$ is a minimal separator and $\Omega$ is a maximal clique) of $H$.
\end{enumerate}
\end{definition}

Observe that the two variants of compatibility differ by parameter $\Omega$ and the last two conditions.
We denote by $\alpha(S,C,W)$ (resp. $\beta(S,C,\Omega,W)$) the size of a largest partial solution compatible with $(S,C,W)$ (resp. $(S,C,\Omega,W)$). We now show how these quantities can be computed over all blocks and all good triples. The dynamic programming will proceed by increasing size over the blocks $(S,C)$, the size of the block being $|S \cup C|$.

It is based on the following equations (see~\cite{FoVi10,FTV15} for details and proofs and Figure~\ref{fi:alphabeta} for an illustration).

\paragraph{Base case.} It occurs for good triples $(S,C,\Omega)$ such that $\Omega = S \cup C$. In this case, for each subset $W$ of $\Omega$ of size at most $t+1$,

\begin{equation}\label{eq:base}
\beta(S,C,\Omega,W) = |W|.
\end{equation}

\begin{figure}[h]
\begin{center}
\includegraphics[scale=0.5]{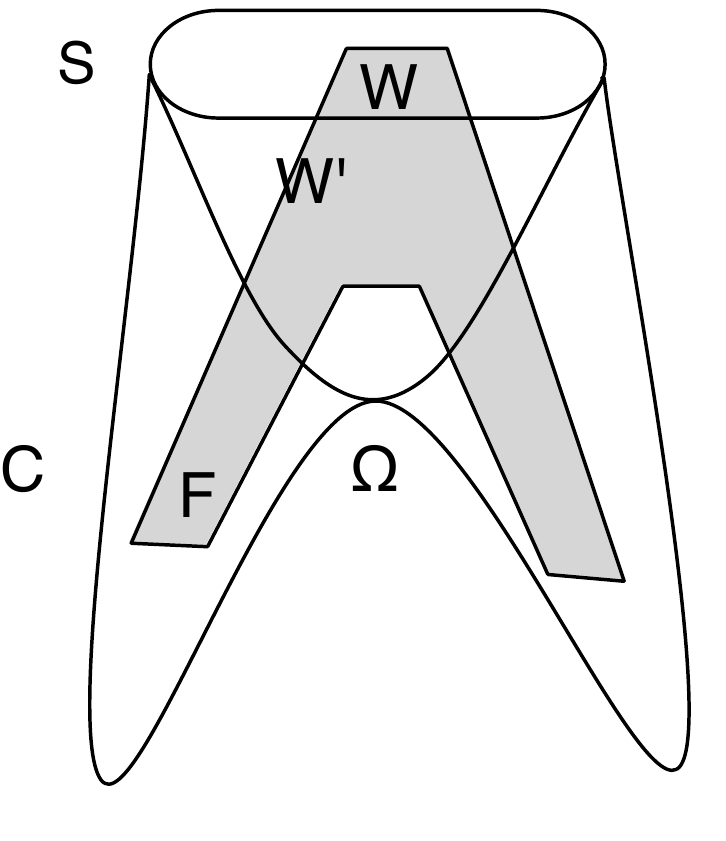}\hspace{1cm}
\includegraphics[scale=0.5]{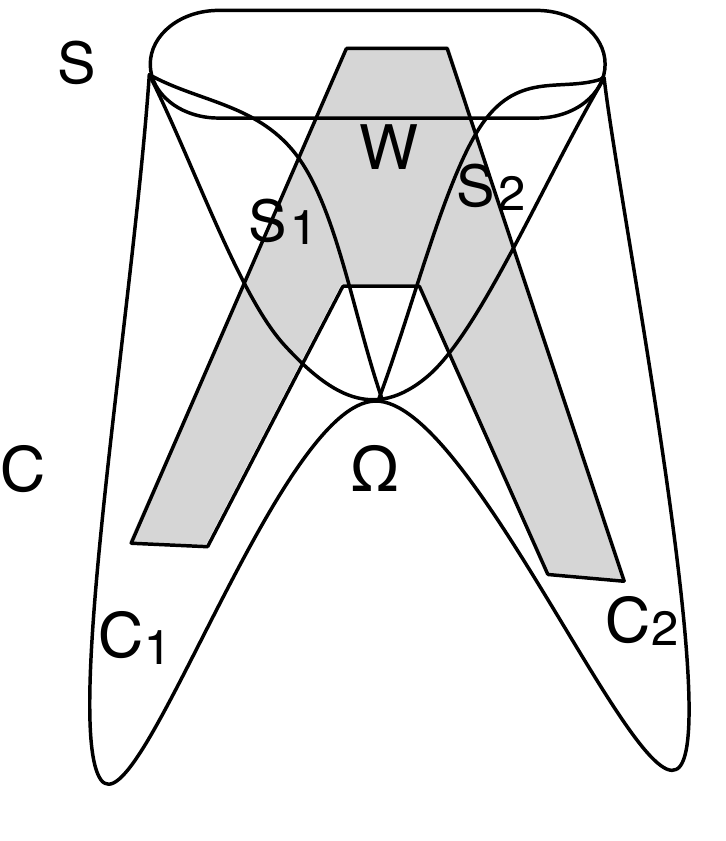}
\end{center}
\vspace{-1cm}
\caption{Computing $\alpha$ form $\beta$ (left), and $\beta$ from $\alpha$ (right).} 
\label{fi:alphabeta}
\end{figure}

\paragraph{Computing $\alpha$ from $\beta$.} The following equation allows to compute the $\alpha$ values from $\beta$ values:

\begin{equation}\label{eq:afb}
\alpha(S,C,W) = \max_{\Omega, W'} \beta(S,C,\Omega,W'),
\end{equation}
where the maximum is taken over all \pmc s $\Omega$ such that $(S,C,\Omega)$ is a good triple, and all subsets $W'$ of $\Omega$, of size at most $t+1$, such that $W = W' \cap S$.

\paragraph{Computing $\beta$ from $\alpha$.} Let $(S,C,\Omega)$ be a good triple, and fix an order $C_1, C_2,\dots,C_p$ on the connected components of $G[C \setminus \Omega]$. Let $S_i = N_G(C_i)$, for all $1 \leq i \leq p$. By~\cite{BoTo01}, $(S_i,C_i)$ are also blocks of $G$. 

A partial solution $F$ compatible with $(S,C,\Omega,W)$ is obtained as a union of partial solutions $F_i$ compatible with $(S_i,C_i,W \cap S_i)$, for each $1 \leq i \leq p$, and the set $W$.
 
Denote by $\gamma_i(S,C,\Omega,W)$ the size of the largest partial solution $F$ compatible\footnote{To be precise, the $\gamma$ function is not required at this stage, if we only compute largest induced subgraphs of treewidth at most $t$. However it becomes necessary when we request the solution to satisfy additional properties, as it will happen in Section~\ref{se:cvc}.} with $(S,C,\Omega,W)$, contained in $\Omega \cup C_1 \cup \dots \cup C_i$ (hence $F$ is not allowed to intersect the components $C_{i+1}$ to $C_p$).

We have the following equations. 

\begin{equation}\label{eq:g1}
\gamma_1(S,C,\Omega,W) = \alpha(S,C,\Omega,W \cap S_1) + |W| - |W \cap S_1|.
\end{equation}

For all $i$, $2 \leq i \leq p$,
\begin{equation}\label{eq:gi}
\gamma_i(S,C,\Omega,W) = \gamma_{i-1}(S,C,\Omega,W) + \alpha(S,C,\Omega,W \cap S_i) - |W \cap S_i|.
\end{equation}

and finally

\begin{equation}\label{eq:bfg}
\beta(S,C,\Omega,W) = \gamma_p(S,C,\Omega,W).
\end{equation}

For convenience we also consider that $\emptyset$ is a minimal separator, and $(\emptyset,V)$ is a block. Then the size of the optimal global solution is simply $\alpha(\emptyset,V,\emptyset)$. The algorithm can be adapted to output an optimal solution, not only its size.

\section{Powers of graphs with polynomially many minimal separators}\label{se:pow}

Let us prove that for any odd $k$, graph $G^k$ has no more minimal separators than $G$. 

\begin{theorem}\label{th:main}
Consider a graph $G$, an odd number $k=2l+1$ with $l\geq 0$, and a minimal separator $\overline{S}$ of $G^k$. Then there exists a minimal separator $S$ of $G$ such that $\overline{S}  = N^{l}_G[S]$.
\end{theorem}

\begin{proof}

The lemma is trivially true if $\overline{S} = \emptyset$. Let $a,b \in V$ such that $\overline{S} \neq \emptyset$ is an $a,b$-minimal separator in $G^k$, and call $C_a$, $C_b$ the components of $G^k - \overline{S}$ that contain $a$ and $b$, respectively. 
Let us call $D_a = N^{l}_G[C_a]$ and $D_b = N^{l}_G[C_b]$.

\smallskip
{\bf Claim 1: $dist_G(D_a, D_b) \geq 2$.} 

Suppose that $dist_G(D_a, D_b) < 2$, and pick $x\in D_a, y \in D_b$ with $dist_G(x,y)\leq 1$ (notice that possibly $x=y$). Let $x_a \in C_a$ and $x_b \in C_b$ be such that there exists an $x_a, x$-path and a $y, x_b$ -path in $G$, each one of length at most $l$, called $P_a$ and $P_b$, respectively. This implies that there must be a $x_a, x_b$-path of length at most $2l + 1 = k$ in $G$, which means that $\{x_a, x_b\} \in E(G^k)$, a contradiction with the fact that $\overline{S}$ separates $C_a$ from $C_b$ in $G^k$.  

\smallskip

{\bf Claim 2}: $\tilde{S} = \overline{S} ~\backslash \left( D_a \cup D_b \right)$ separates $a$ and $b$ in $G$. 

Notice first that $N_G(D_a) \subseteq N_G^{l+1}(C_a) \subseteq \overline{S}$. Suppose that $\tilde{S}$ does not separate $a$ and $b$, and let $P$ be an $a,b$-path in $G$ that does not pass through $\tilde{S}$. Let $x_1, \dots, x_{s-2}$ the internal nodes of $P$, where $s = |P|$, and consider $i = \max\{j  ~|~ x_j \in D_a\cap P\}$. Since $P \cap \tilde{S} = \emptyset$, necessarily $x_{i+1} \in  D_b$,  a contradiction with Claim 1.


\smallskip

{\bf Claim 3}: $D_a$ and $D_b$ are connected subsets of $G$. 

This is straightforward from the definition of the sets, $D_a = N^{l}_G[C_a]$ and $D_b = N^{l}_G[C_b]$, and the fact that $C_a$ and $C_b$ are connected in $G$.

\medskip

Let $\tilde{C}_b$ be the connected component of $G - N_G[D_a]$ that contains $b$, and denote $S = N_G(\tilde{C_b})$. Note that $S \subseteq \tilde{S} \subset \overline{S}$. By applying Proposition \ref{prop:min_sep}, we have  that $S$ is a minimal $a,b$-separator in $G$. Call $\tilde{C}_a$ the component of $G - S$ that contains $a$. Since $S \subseteq \tilde{S}$, we have that $D_b \subseteq \tilde{C}_b$ and $D_a \subseteq \tilde{C}_a$.

\smallskip

{\bf Claim 4}: $N^l_G[S] = \overline{S}$. 

We first prove that $N^l_G[S] \subseteq \overline{S}$. By construction, $S \subseteq N_G(D_a)$. Consequently $S\subseteq N_G^{l+1}(C_a) \backslash N_G^{l}(C_a)$, therefore $N_G^l[S] \subseteq N_G^{2l+1}(C_a) = N_{G^k}(C_a) = \overline{S}$. 

Conversely, we must show that every vertex $x$ of $\overline{S}$ is in $N^l_G[S]$. 
%
By contradiction, let $x \in \overline{S} \setminus N^l_G[S]$. We distinguish two cases~: $x \in \tilde{C}_a$, and $x \in \overline{S} \setminus \tilde{C}_a$.
In the first case,  since $N_{G^k}(C_b) = \overline{S}$, there exists a path from some vertex $y\in C_b$ to $x$ of length at most $k$, in graph $G$. Let us call $P$ one of those $y,x$-paths. Observe that the first $l+1$ vertices of the path belong to $D_b \subseteq \tilde{C_b}$, and none of the last $l+1$ vertices of the path belongs to $S$ (otherwise $x \in N_G^l[S]$). Then $P$ is a path that connects $\tilde{C_b}$ with $\tilde{C_a}$ without passing through $S$, a contradiction with the fact that $S$ separates $a$ and $b$ in graph $G$. 

It remains to prove the last case, when $x\in \overline{S} \setminus \tilde{C}_a$. Since $N_{G^k}(C_a) = \overline{S}$, there exists a node $y \in C_a$ such that there is a $y,x$-path $P$ of length at most $k$ in $G$. Since the first $l+1$ vertices of the path belong to $D_a$, and the last $l+1$ vertices of the path do not belong to $S$, we deduce that $P$ is an $y,x$-path in $G$ that does not intersect $S$. The path can be extended (through $C_a$) into an $a,x$-path that does not intersect $S$, a contradiction with the fact that $x$ does not belong to $\tilde{C}_a$.  This concludes the proof of our theorem. 
\qed
\end{proof}

Recall that \textsc{Distance-$d$ Independent set} on $G$ is equivalent to \textsc{Maximum Independent Set} on $G^{d-1}$. Since the latter problem is polynomial on $\Gpoly$ by Proposition~\ref{pr:FoVi10}, we deduce:
\begin{theorem}\label{th:odddis}
For any even value $d$, and any polynomial $\poly$, problem \textsc{Distance-$d$ Independent set} is polynomially solvable on $\Gpoly$.
\end{theorem}

We remind that for any odd value $d$,  problem \textsc{Distance-$d$ Independent set} is NP-hard on chordal graphs~\cite{EGM14}, thus on $\Gpoly$ for any polynomial $\poly$ asymptotically larger than $n$. The construction of~\cite{EGM14} also shows that even powers of chordal graphs may contain exponentially many minimal separators. 

\section{On \textsc{Connected Vertex Cover} and \textsc{Connected Feedback Vertex Set}}\label{se:cvc}

Let us consider the problem of finding a maximum induced subgraph $G[F]$ such that $\tw(G[F])\leq t$ and $G - F$ is connected. One can easily observe that, for $t=0$ (resp. $t=1$), this problem is equivalent to \textsc{Connected Vertex Cover} (resp. \textsc{Connected Feedback Vertex Set}), in the sense that if $F$ is an optimal solution for the former, than $V(G) - F$ is an optimal solution for the latter. 

Our goal is to enrich the dynamic programming scheme described in Subsection~\ref{ss:FoVi10} in order to ensure the connectivity of $G - F$. One should think of this dynamic programming scheme of Subsection~\ref{ss:FoVi10} as similar to dynamic programming algorithms for bounded treewidth. The difference is that the bags (here, the \pmc s) are not small but polynomially many, and we parse simultaneously through a set of decompositions. Nevertheless, we can borrow several classical ideas from treewidth-based algorithms.

In general, for checking some property for the solution $F$, we add a notion of \emph{characteristics} of partial solutions. Then, for a characteristic $c$, we update the Definition~\ref{de:partcomp} in order to define partial solutions compatible with $(S,C,W,c)$ (resp. $(S,C,\Omega,W,c)$), by requesting the partial solution to be compatible with characteristic $c$. Parameter $c$ will also appear in the updated version of Equations~\ref{eq:base} to~\ref{eq:bfg}. 

As usual in dynamic programming, the characteristics must satisfy several properties: (1) we must be able to compute the characteristic for the base case, (2) the characteristic of a partial solution $F$ obtained from gluing smaller partial solutions $F_i$ must only depend on the characteristics of $F_i$, and (3) the characteristic of a global solution should indicate whether it is acceptable or not. Moreover, for a polynomial algorithm, we need the set of possible characteristics to be polynomially bounded.

For checking connectivity conditions on $G-F$, we define the characteristics of partial solutions in a natural way. Consider a block $(S,C)$ (resp. a good triple $(S,C,\Omega)$) and a subset $W$ of 
$S$ (resp. of $\Omega$). 
Let $F$ be a partial solution compatible with $(S,C,W)$ (resp. $(S,C,\Omega,W)$), see Definition~\ref{de:partcomp}. The \emph{characteristic} $c$ of $F$ for $(S,C,W)$ (resp. for $(S,C,\Omega,W)$) is defined as the partition induced on $S \setminus W$ (resp. on $\Omega \setminus W$) by the connected components of  $G[S \cup C] - F$. More formally, let $D_1, \dots, D_q$ denote the connected components of $G[S \cup C] - F$, and let $P_j = D_j \cap S$ (resp. $P_j = D_j \cap \Omega$), for all $1 \leq j \leq q$. Then $c = \{P_1,\dots, P_q\}$. We decide that if $S \neq \emptyset$, partial solutions $F$ having some component $D_j$ that does not intersect $S$ (resp. $\Omega$) are immediately rejected; indeed, for any extension $F'$ of $F$, the graph $G - F'$ remains disconnected. Hence we may assume that all sets $P_j$ are non-empty.

We say that a partial solution $F$ is \emph{compatible with $(S,C,W,c)$} (resp. \emph{with $(S,C,\Omega,W,c)$}) if it satisfies the conditions of Defintion~\ref{de:partcomp}, and $c$ is the characteristic of $F$ for $(S,C,W)$ (resp. for $(S,C,\Omega,W)$).

We also define functions $\alpha(S,C,W,c)$, $\beta(S,C,\Omega,W,c)$ and $\gamma_i(S,C,\Omega,W,c)$ like in Subsection~\ref{ss:FoVi10}, as the maximum size of partial solutions $F$ compatible with the parameters. 
We can update Equations~\ref{eq:base} to~\ref{eq:bfg} as follows.

\paragraph{Base case.} For the good triples $(S,C,\Omega)$ such that $(S,C)$ is inclusion-minimal (hence $\Omega = S \cup C$),
\begin{equation}\label{eq:base2}
\beta(S,C,\Omega,W,c) = |W|~\text{if $c$ corresponds to the connected components of $G[\Omega \setminus W]$.}
\end{equation}
Otherwise we set $\beta(S,C,\Omega,W,c) = -\infty$.

\paragraph{Computing $\alpha$ from $\beta$.} 
\begin{equation}\label{eq:afb2}
\alpha(S,C,W,c) = \max_{\Omega, W',c'} \beta(S,C,\Omega,W',c'),
\end{equation}
where the maximum is taken over all \pmc s $\Omega$ such that $(S,C,\Omega)$ is a good triple, and all subsets $W'$ of $\Omega$, of size at most $t+1$, such that $W = W' \cap S$, and all characteristics $c'$ such that each part of $c$ corresponds to the intersection between $S$ and a part of $c'$.
If $S \neq \emptyset$ we also request that each part of $c'$ intersects $S$. This condition allows to reject partial solutions $F$ for which a component of $G[S \cup C] - F$ is strictly contained in $C$. Indeed, such partial solutions cannot extend to global solutions $F'$ such that $G - F'$ is connected.

For the particular case $S = \emptyset$ (hence $C=V$ and $W = \emptyset$) we only consider characteristics $c'$ with a single part. This ensures that the global solution $F$ satisfies that $G - F$ is connected.

If there is no such triple $(\Omega, W',c')$, then we set $\alpha(S,C,W,c) = - \infty$ (we can assume that, when it has no parameters, function $\max$ returns $-\infty$).

\paragraph{Computing $\beta$ from $\alpha$.} 

\begin{equation}\label{eq:g12}
\gamma_1(S,C,\Omega,W,c) = \max_{c'} (\alpha(S_1,C_1,\Omega,W \cap S_1, c') + |W| - |W \cap S_1|),
\end{equation}
over all characteristics if $c'$ that \emph{map correctly} on $c$, in the following sense. Consider a characteristic $c'$ and let $G_{c'}[\Omega \setminus W]$ be the graph obtained from $G[\Omega \setminus W]$ by completing each part $D \in c'$ into a clique. We say that a characteristic $c'$ maps correctly on $c$ if $c$ is the partition of $\Omega \setminus W$ defined by the connected components of $G_{c'}[\Omega \setminus W]$.

The notion of mapping transforms the characteristic of the partial solution $F_1$ w.r.t. $(S_1,C_1,W \cap S_1)$ into the characteristic of $F_1 \cup W$ w.r.t. the quadruple $(S,C,\Omega,W)$.

For all $i$, $2 \leq i \leq p$,
\begin{equation}\label{eq:gi2}
\gamma_i(S,C,\Omega,W,c) = \max_{c_{i-1},c_i}(\gamma_{i-1}(S,C,\Omega,W,c_{i-1}) + \alpha(S,C,\Omega,W \cap S_i,c_i) - |W \cap S_i|),
\end{equation}
over all pairs of characteristics $c_{i-1},c_i$ that \emph{map correctly} on $c$. That is, $c$ must correspond to the connected components of $G_{c_{i-1},c_i}[\Omega \setminus W]$, obtained from $G[\Omega \setminus W]$ by completing each part of $c_{i-1}$ and each part $c_i$ of into a clique.

Finally
\begin{equation}\label{eq:bfg2}
\beta(S,C,\Omega,W,c) = \gamma_p(S,C,\Omega,W,c).
\end{equation}

The optimal solution size is, as before, $\alpha(\emptyset,V,\emptyset,\{\emptyset\})$.

In general, the number of characteristics may be exponential. Nevertheless, there are classes of graphs with the property that each minimal separator $S$ and each \pmc\ $\Omega$ can be partitioned into at most a constant number of cliques. With this constraint, the number of characteristics is polynomial (even constant, for any given triple $(S,C,W)$ or quadruple $(S,C,\Omega,W)$).

This is the case for chordal graphs, where each minimal separator and each \pmc\ induces a clique in $G$.

It is also the case for circular-arc graphs. Recall that each minimal separator corresponds to the set of arcs intersecting a pair of scanpoints~\cite{KKW98}. Moreover, by~\cite{KKW98,BoTo01}, each \pmc\ corresponds to the set of arcs intersecting a triple of scanpoints. Since arcs intersecting a given scanpoint form a clique, we have that each minimal separator can be partitioned into  two cliques, and each \pmc\ can be partitioned into three cliques. 

We deduce:
\begin{theorem}\label{th:cvc}
On chordal and circular-arc graphs, problems \textsc{Connected Vertex Cover} and~\textsc{Connected Feedback Vertex Set} are solvable in polynomial time. More generally, one can compute in polynomial time a maximum vertex subset $F$ such that $G[F]$ is of treewidth at most $t$ and $G - F$ is connected.
\end{theorem}

Note that Escoffier \textit{et al.}~\cite{EGM10} already observed that \textsc{Connected Vertex Cover} is polynomial for chordal graphs.

\section{\textsc{Independent Dominating Set} and variants}\label{se:ind}

The \textsc{Independent Dominating Set} problem consists in finding a \emph{minimum} independent set $F$ of $G$ such that $F$ dominates $G$. Hence the solution $F$ induces a graph of treewidth $0$ and it is natural to ask if similar techniques work in this case. The fact that we have a minimization problem is not a difficulty: the general dynamic programming scheme applies in this case, and for any weighted problem with polynomially bounded weights, including negative ones~\cite{FoVi10,FTV15}.

\textsc{Independent Dominating Set} is known to be NP-complete in chordal bipartite graphs~\cite{DMK90} and in circle graphs~\cite{BGMPST14}. Therefore, it is NP-hard on $\Gpoly$ for some polynomials $\poly$. But, again, we can use our scheme in the case of circular-arc graphs, for this problem or any problem of the type minimum dominating induced subgraph of treewidth at most a constant~$t$.

Let $(S,C)$ be a block an let $F \subseteq S \cup C$ be a partial solution compatible with $(S,C,W)$ for some $W \subseteq S$ of size at most $t+1$ (in the sense of Definition~\ref{de:partcomp}). The natural way for defining the characteristic of $F$ is to specify which vertices of $S$ are dominated by $F$ and which are not (we already know that $F \cap S = W$). It is thus enough to memorize which vertices of $S$ are dominated by $F \cap C$. In circular-arc graphs, this information can be encoded using a polynomial number of characteristics. Indeed, a minimal separator $S$ corresponds to arcs intersecting a scanline, between two scanpoints $p_1$ and $p_2$ of some intersection model of $G$. Moreover (see~\cite{KKW98}), the vertices of component $C$ correspond to the arcs situated on one of the sides of the scanline. Let $s^1_1,s^1_2,\dots, s^1_{l_1}$ be the arcs of the model containing scanpoint $p_1$, ordered by increasing intersection with the side of $p_1p_2$ corresponding to $C$. Simply observe that if $F \cap C$ dominates vertex  $s^1_i$, it also dominates all vertices $s^1_j$ with $j>i$. Therefore we only have to store the vertex $s^1_{min_1}$ dominated by $F \cap C$ which has a minimum intersection with the side of the scanline corresponding to component $C$, and proceed similary for the arcs of $S$ containing scanpoint $p_2$. These two vertices of $S$ will define the characteristic of $F$, and they suffice to identify all vertices of $S$ dominated by $F \cap C$.

\begin{figure}
\centering
\begin{tikzpicture}[scale=1] 

\draw(-1,0) arc (0:360:2) ;
\draw[ dotted] (-0.66,0) arc (0:10:2.33);
\draw[-|] (-0.66,0) arc (0:-30:2.33);
\draw[ dotted] (-0.33,0) arc (0:10:2.66);
\draw[-|] (-0.33,0) arc (0:-45:2.66) ;
\draw[dotted] (0,0) arc (0:10:3);
\draw[-|] (0,0) arc (0:-60:3);

\draw[-|] (-6,0) arc (180:230:3) ;
\draw[dotted] (-6,0) arc (180:170:3) ;

\draw[-|] (-5.66,0) arc (180:200:2.66) ;
\draw[dotted] (-5.66,0) arc (180:170:2.66) ;

\draw[densely dotted, thick, |-|] (-4.8,-1.5) arc (220:310:2.33);

\node (p1) at (-0.75,-0.3){$\phantom{p_1}$};

\node (p2) at (-5.25,-0.3){$\phantom{p_2}$};

\draw[-,very thick] (p1)--(p2);

\node (pp1) at (-1.2,-0.1){$p_1$};
\node (pp2) at (-4.75,-0.1){$p_2$};

\node (C) at (-3.3,-2.6){$F \cap C$};

\node (pp2) at (-1.1,-1.35){$s_1^1$};
\node (pp2) at (-1.2,-2.1){$s_2^1$};
\node (pp2) at (-1.7,-2.6){$s_3^1$};

\node (pp2) at (-5.4,-1.1){$s_1^2$};
\node (pp2) at (-4.75,-2.4){$s_2^2$};

\end{tikzpicture} 
\caption{Domination in circular-arc graphs. The characteristic of $F$ w.r.t. $(S,C)$ is $(s_3^1, s_2^2)$.}\label{fi:ca}
\end{figure}
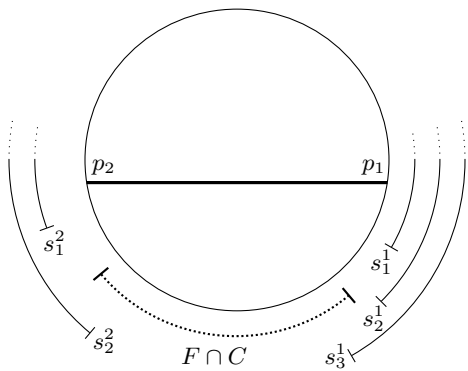

These characteristics can be used to compute a minimum dominating induced subgraph of treewidth at most $t$, for circular-arc graphs, in polynomial time. We will not show, in details, how to do it, since the technique is quite classical. Problem \textsc{Independent Dominating Set} is already known to be polynomial for this class~\cite{Ch98,Va12}. The algorithm of Vatshelle~\cite{Va12} is more general,  based on parameters called \emph{boolean-width} and \emph{MIM-width}, which are small ($\cO(\log n)$ for the former, constant for the latter) on circular-arc graphs and also other graph classes. 



Another problem of similar flavor, combining domination and independence, is \textsc{Red-Blue Dominating Set}. In this problem we are given a bipartite graph $G=(R,B,E)$ with red and blue vertices,  and an integer $k$, and the goal is to find a set of at most $k$ blue vertices dominating all the red ones. \textsc{Red-Blue Dominating Set} can be reduced to \textsc{Connected Vertex Cover} as follows~\cite{DLS14}. Let $G'$ be the graph obtained from $G=(R,B,E)$ by adding a new vertex $u$ adjacent to all vertices of $B$ and then, for each $v \in R \cup \{u\}$, a pendant vertex $v'$ adjacent only to $v$. Then $G$ has a red-blue dominating set of size at most $k$ if and only if $G'$ has a connected vertex cover of size at most $k+|B|+1$. Indeed any minimum connected vertex cover of $G'$ must contain $u$, $R$, and a subset of $B$ dominating $R$. It is not hard to prove that this reduction increases the number of minimal separators by at most $\cO(n)$, see Appendix~\ref{app:red}.

Therefore, if \textsc{Red-Blue Dominating Set} is NP-hard on (bipartite) $\Gpoly$ for some $poly$, so is  
\textsc{Connected Vertex Cover}. 
There are two natural, well-studied classes of bipartite graphs with polynomial number of minimal separators, and it turns out that \textsc{Red-Blue Dominating Set} is polynomial for both. One is the class of chordal bipartite graphs (which are actually defined as the bipartite, \emph{weakly} chordal graphs). For this class, \textsc{Red-Blue Dominating Set} is polynomial by~\cite{DMK90}. Reference~\cite{DMK90}  considers the total domination problem for the class, but the approach is based on red-blue domination.


The second natural class is the class of  circle bipartite graphs, i.e., bipartite graphs that are also circle graphs. They have an elegant characterization established by de Fraysseix~\cite{Fr81}. Let $H = (V,E)$ be a planar multigraph, and partition its edge set into two parts $E_R$ and $E_B$ such that $T=(V,E_R)$ is a spanning tree of $H$. Let $B(H, E_R) = (E_R, E_B, E')$ be the bipartite graph defined as follows: $E_R$ is the set of red vertices, $E_B$ is the set of blue vertices, and $e_R \in E_R$ is adjacent to $e_B \in E_B$ if the unique cycle obtained from the spanning tree $T$ by adding $e_B$ contains the edge $e_R$.  We say that $B(H, E_R)$ is a fundamental graph of $H$. By~\cite{Fr81},  a graph is circle bipartite if and only if it is the fundamental graph $B(H, E_R)$  of a planar multigraph $H$.


Consider now the \textsc{Tree augmentation} problem that consists in finding, on input  $G$ and a spanning tree $T$ of $G$, a minimum set of edges $D \subseteq E(G)-E(T)$ such that each edge in $E(T)$ is contained in at least one cycle of $G' = (V, E(T) \cup D)$.  In \cite{Provan199987} is shown that \textsc{Tree augmentation} is polynomial when the input graph is planar.  Is direct to see that a set $S \subseteq E_B$ is a solution of the  \textsc{Tree augmentation} problem on input $H = (V, E_R \cup E_B)$ and $T = (V, E_R)$,  if and only if $S$ is a solution of  \textsc{Red-Blue Dominating Set} on input $B(H) = (E_R, E_B, E')$. This observation, together with \cite{Fr81} and \cite{Provan199987}, impliy that \textsc{Red-Blue Dominating Set} is polynomial in circle bipartite graphs. 


\section{Discussion}

We showed how the dynamic programming scheme of~\cite{FoVi10,FTV15} can be extended for other optimization problems, on \emph{subclasses} of $\Gpoly$. Note that the algorithm of~\cite{FTV15} allows to find in polynomial time, on $\Gpoly$, a maximum (weight) subgraph $G[F]$ of treewidth at most $t$, satisfying some property expressible in CMSO. It also handles annotated versions, where the vertices/edges of $G[F]$ must be selected from a prescribed set. 

We have seen that \textsc{Distance-$d$ Independent Set} can be  solved in polynomial time on $\Gpoly$ for any even $d$. This also holds for the more general problem of finding an induced subgraph $G[F]$ whose components are at pairwise distance at least $d$, and such that each component is isomorphic to a graph in a fixed family. E.g., each component could be an edge, to have a variant of \textsc{Maximum Induced Matching} where edges should be at pairwise distance at least $d$. For this we need to solve the corresponding problem on $G^{d-1}$, using only edges from $G$, as in~\cite{FTV15}.

When seeking for maximum (resp. minimum) induced subgraphs $G[F]$ of treewidth at most $t$ such that $G - F$ is connected (resp. $F$ dominates $G$) on particular subclasses of $\Gpoly$, we can add any CMSO condition on $G[F]$. It is not unlikely that the techniques can be extended to other classes than circular-arc graphs (and chordal graphs, for connectivity constraints).

We also believe that the interplay between graphs of bounded MIM-width~\cite{Va12} and $\Gpoly$ deserves to be studied. None of the classes contains the other, but several natural  graph classes are in their intersection, and they are both somehow related to induced matchings.

We leave as open problems the complexity of \textsc{Connected Vertex Cover} and \textsc{Connected Feedback Vertex set} in weakly chordal graphs, and on $\Gpoly$. We have examples showing that, even for weakly chordal graphs, the natural set of characteristics that we used in Section~\ref{se:cvc} is not polynomially bounded.

\paragraph{Acknowledgements.} We  thank Iyad Kanj for fruitful discussions on the subject.
\bibliographystyle{plain}
\bibliography{OptMinSep}

\appendix

\section{Reduction from \textsc{Red-Blue Dominating Set} to \textsc{Connected Vertex Cover} on bipartite $\Gpoly$}\label{app:red}

\begin{lemma}
Let $G'$ be the graph obtained from the bipartite graph $G=(R,B,E)$ by adding a new vertex $u$ adjacent to all vertices of $B$ and then, for each $v \in R \cup \{u\}$, a pendant vertex $v'$ adjacent only to $v$.

$G'$ has at most $|V(G')|$ more minimal separators than $G$. 
\end{lemma}
\begin{proof}
Let $S$ be a minimal separator of $G'$ and let $C,D$ be two components of $G' - S$ such that $N_{G'}(C) = N_{G'}(D) = S$. 

Assume first that there are two blue vertices $c \in C$ and $d \in D$. Then vertex $u$ (seeing the whole set $B$) must be in $S$. No pendant vertex can be in $S$ because vertices of $S$ have neighbors in both $C$ and $D$. Therefore $S \setminus \{u\}$ also separates $c$ and $d$ in $G$, it is even a $c,d$-minimal separator. Indeed, by removing the pendant vertices from $C$ and $D$, we obtain two components of $G - S$ whose neighborhood in graph $G$ is $S  \setminus \{u\}$. Therefore the number of minimal separators $S$ of $G'$ of this type is at most the number of minimal separators of $G$.

It remains to consider the case when one of the components, say $C$, does not contain any blue vertex. Then $C$ contains at most two vertices. Actually, $C$ either contains a unique pendant vertex, or a vertex $v \in R \cup \{u\}$ and its pendant vertex ($C$ cannot contain a unique vertex $v \in R \cup \{u\}$, because its pendant neighbor should be in $S$, or $S$ cannot contain pendent vertices). 
Hence the number of minimal separators of this type is at most $|V(G')|$.
\qed
\end{proof}

\end{document}